\documentclass{article}

\PassOptionsToPackage{numbers}{natbib}
\usepackage[final]{neurips_2020}

\usepackage{xcolor}
\usepackage{enumitem}
\usepackage{url}            
\usepackage{booktabs}       
\usepackage{amsfonts}       
\usepackage{nicefrac}       
\usepackage{microtype}      
\usepackage{graphicx}
\usepackage{array,multirow,graphicx}
\usepackage{amssymb}
\usepackage{amsmath}
\usepackage{amsthm}
\usepackage{mathtools}
\usepackage{algorithm}
\usepackage{algorithmic}
\usepackage{setspace}
\usepackage{caption}
\usepackage{subcaption}
\usepackage{xspace}
\usepackage{multirow}
\usepackage[export]{adjustbox}
\usepackage{thmtools}
\usepackage{thm-restate}

\definecolor{darkgray}{gray}{0.25}
\definecolor{lightblue}{rgb}{0.25,0.25,1}
\usepackage[colorlinks,citecolor=darkgray,urlcolor=lightblue]{hyperref}

\newcommand{\bS}{\mathbb{S}}
\newcommand{\ind}{\mathbf{1}}
\newcommand{\reals}{\mathbb{R}}

\newcommand{\naturals}{\mathbb{N}}

\newcommand{\M}{\mathcal{M}}
\newcommand{\N}{\mathcal{N}}

\newcommand{\cS}{\mathcal{S}}
\newcommand{\T}{\mathcal{T}}

\newcommand{\Y}{\mathcal{Y}}

\DeclareMathOperator*{\argmin}{argmin}

\newtheorem{definition}{Definition}

\newtheorem{lemma}{Lemma}

\newtheorem{proposition}{Proposition}

\newcommand{\seclabel}[1]{\label{sec:#1}}
\newcommand{\secref}[1]{Section~\ref{sec:#1}}

\newcommand{\figlabel}[1]{\label{fig:#1}}
\newcommand{\figref}[1]{Figure~\ref{fig:#1}}

\newcommand{\eqlabel}[1]{\label{eq:#1}}
\renewcommand{\eqref}[1]{Equation~\ref{eq:#1}}

\newcommand{\deflabel}[1]{\label{def:#1}}
\newcommand{\defref}[1]{Definition~\ref{def:#1}}

\newcommand{\lemlabel}[1]{\label{lem:#1}}
\newcommand{\lemref}[1]{Lemma~\ref{lem:#1}}

\newcommand{\thmlabel}[1]{\label{thm:#1}}
\newcommand{\thmref}[1]{Theorem~\ref{thm:#1}}

\newcommand{\algmlabel}[1]{\label{alg:#1}}
\newcommand{\algmref}[1]{Algorithm~\ref{alg:#1}}

\newcommand{\proplabel}[1]{\label{prop:#1}}
\newcommand{\propref}[1]{Proposition~\ref{prop:#1}}

\renewcommand{\output}{[a,b]}  

\newcommand{\ouru}{\textit{Duff}\xspace}

\newcommand{\uadmis}{u_d}

\newcommand{\ssmech}{\textsc{SmoothSens}\xspace}
\newcommand{\ssmeche}{\textsc{SmoothSens}$_{\epsilon}$\xspace}
\newcommand{\ssmeched}{\textsc{SmoothSens}$_{\epsilon,\delta}$\xspace}

\newcommand{\loglap}{\textsc{LaplaceLN}\xspace}

\newcommand{\sensu}{\Delta_u}

\newcommand{\GS}{\mathsf{GS}}
\renewcommand{\SS}{\mathsf{SS}}
\newcommand{\LS}{\mathsf{LS}}

\begin{document}

\title{\ouru: A {\it D}ataset-Distance-Based {\it U}tility {\it F}unction {\it F}amily for the Exponential Mechanism}

\author{%
  Jennifer Gillenwater \\
  Google Research \\
  111 8th Ave, New York, NY \\
  \texttt{jengi@google.com}
  \And
  Andr\'es Mu\~noz Medina \\
  Google Research \\
  111 8th Ave, New York, NY \\
  \texttt{ammedina@google.com}
 }

\maketitle

\begin{abstract}
We propose and analyze a general-purpose {\it d}ataset-distance-based {\it
  u}tility {\it f}unction {\it f}amily (\ouru) for differential privacy's
exponential mechanism.  Given a particular dataset and a statistic (e.g.,
median, mode), this function family assigns utility to a possible output $o$
based on the number of individuals whose data would have to be added to or
removed from the dataset in order for the statistic to take on value $o$. We
show that the exponential mechanism based on \ouru often offers provably higher
fidelity to the statistic's true value compared to existing differential privacy
mechanisms based on smooth sensitivity. In particular, \ouru is an affirmative
answer to the open question of whether it is possible to have a noise
distribution whose variance is proportional to smooth sensitivity and whose
tails decay at a faster-than-polynomial rate. We conclude our paper with an
empirical evaluation of the practical advantages of \ouru for the task of
computing medians.
\end{abstract}

\section{Introduction}
\seclabel{introduction}
The age of big data has brought with it an unprecedented ease-of-access to large
amounts of information.  From predicting traffic in cities to explaining voting
patterns across the country, readily available large datasets have made it
possible for researchers to analyze and model the behavior of people at a very
fine granularity.  Unrestricted access to this data, however, puts the privacy
of individuals at risk. Indeed, many datasets used today contain very sensitive
information about individuals, such as one's political affiliation or home
location. For this reason, analysts and machine learning practitioners must
ensure that the output of their research (models or aggregated statistics)
provide only information about the population as a whole and do not leak
personal details. One might feel, intuitively, that certain outputs are
inherently ``safe''.  For example, statistics aggregated over a sufficiently
large number of individuals, such as the median of a dataset that contains data
from thousands of people.  But in fact it has been proven~\cite{dinur2003pods},
and later verified in practice~\cite{abowd2018kdd}, that even with such
aggregated statistics one can discover very personal information.

A formal guarantee, called \emph{differential privacy}
(DP)~\cite{dwork2006icalp}, has become the standard way to ensure that
individual information is protected.  Examples of its use in industry and by
government entities can be found in Google~\cite{erlingsson2014ccs},
Apple~\cite{pease2016wwdc}, and the census bureau~\cite{abowd2018kdd}.  In
order to protect user information, a DP \emph{mechanism} usually perturbs the
output of an aggregation with noise. The amount of noise depends on the amount
of privacy we want to provide. There is an inherent trade-off here: more privacy
implies that the perturbed value will likely be further from the true
value. Proximity to the truth is known as the \emph{utility} of the mechanism.
The goal of DP research is to maximize utility while preserving a desired level
of privacy. This process has lead to the creation of a plethora of mechanisms
whose privacy-utility trade-offs are often difficult to compare.  Moreover, even
for the problem of releasing simple statistics such as the median, there is a
lack of thorough empirical evaluation.  This has made adoption of DP by
practitioners difficult.

The goal of this paper is to help simplify practitioners' work by
providing utility guarantees for DP's exponential
mechanism~\cite{mcsherry2007focs} when it is used with a
general-purpose dataset-distance-based utility function family, or
\ouru. \ouru is general enough to be useful for many tasks, but also
specific enough that we can prove meaningful theorems about its
privacy-utility trade-off.  The paper is organized as follows: We
begin by introducing the basic definitions of DP and discussing
previous work. In \secref{ouru} we introduce \ouru and provide an
analysis of its utility. This analysis shows a previously unknown
connection between the exponential mechanism and \emph{smooth
sensitivity}~\cite{nissim2007stoc}. We also instantiate \ouru for
several statistics.  In \secref{algorithms} we present a general
sampling algorithm for an exponential mechanism using \ouru.  Finally,
in \secref{experiments} we empirically evaluate on the task of
computing medians.

\section{Definitions}
\seclabel{def}
In this section we present some common concepts of DP. We will denote by $\bS$ a
universe of datasets. A dataset $\cS \in \bS$ is a collection of information
about individuals.  We assume that each individual contributes no more than one
value to a given $\cS$.
\begin{definition}
  Datasets $\cS, \cS' \in \bS$ are \textbf{neighbors} if $\cS$ can be obtained
  from $\cS'$ by adding or removing a single element.  We denote the neighbors
  of a dataset $\cS$ by $\N(\cS)$.
\end{definition}
\begin{definition}
 Let $\M$ denote a (possibly randomized) function mapping a dataset $\cS$ to an
 output in $\output$.  We say $\M$ is an \textbf{$(\epsilon, \delta)$-DP
 mechanism} if, for any two neighboring datasets $\cS$ and $\cS'$, and any set
 of outcomes $A \subseteq \output$, it holds that: $\Pr(\M(\cS) \in A) \leq
 e^{\epsilon} \Pr(\M(\cS') \in A) + \delta$.
\end{definition}
Given some target statistic of a dataset, $T : \bS \to \output$ (e.g., median,
mode), one common mechanism $\M$ that is often used to release a DP version of
$T$ is $\M(\cS) = T(\cS) + \textrm{noise}$.  The exact distribution and scale of
the noise required to satisfy the DP definition has been the subject of
substantial research. Much of this research has focused on understanding the
so-called \emph{sensitivity} of $T$.
\begin{definition}
\deflabel{local-sens}
Given a $T \colon \bS \to \output$ and an $\cS \in \bS$, the \textbf{local sensitivity} of $T$ at $\cS$ is:
  \begin{equation*}
    \LS(T, \cS) = \max_{\cS' \in \N(\cS)} |T(\cS) - T(\cS')|.
    \end{equation*}
\end{definition}
\begin{definition}
\deflabel{global-sens}
  Given a $T : \bS \to \output$, the \textbf{(global) sensitivity} of
    $T$ is: $\GS(T) = \max_{\cS \in \bS} \LS(T, \cS)$.
\end{definition}
Adding noise to $T$ with variance proportional to $\GS(T)$ is one way of
constructing a mechanism $\M$ that is DP.
\begin{proposition}[\textbf{Laplace mechanism}]
  \proplabel{laplace-mech}
  Let $T$ be a function with sensitivity $\GS(T)$.  Then the mechanism $\M$ that
  releases $\M(S) = T(S) + \frac{\GS(T)}{\epsilon} Z$, where $Z \sim Lap(0,
  1)$, is $(\epsilon, 0)$-DP.
\end{proposition}
While the Laplace mechanism is simple, the amount of noise added can be more
than is strictly necessary.  This is because $\GS(T)$ is a max over
all possible datasets.  It might seem like 
an easy fix is to just add noise proportional to $\LS(\cS)$.  However, this is
not DP, because $\LS(\cS)$ is itself a sensitive quantity; see Section 2.1
of~\cite{nissim2007stoc} for a detailed example.  To bridge the gap between
$\LS$ and $\GS$, \cite{nissim2007stoc} introduced the notion of \emph{smooth
sensitivity}, which depends on the distance between datasets.
\begin{definition}
\deflabel{dataset-distance}
Datasets $\cS, \cS'$ are at \textbf{distance} $k$, denoted $d(\cS, \cS') = k$,
if the existence of a sequence $\cS = \cS_0, \ldots,
\cS_m = \cS'$ such that $\cS_i \in \N(\cS_{i-1})$ implies $m \geq k$.
\end{definition}
For completeness, the appendix contains a proof that $d$ is a true
distance function (a metric) over $\bS$.
\begin{definition}
\deflabel{smooth-sens}
  Given a function $T \colon \bS \rightarrow \output$, $\beta >0$ and a dataset
$\cS \in \bS$, the \textbf{$\beta$-smooth sensitivity} of $T$ at $\cS$ is:
$\SS_\beta(T, \cS) = \max_{k \geq 0} \max_{\cS' \colon d(\cS, \cS')=k} e^{-\beta
k}\LS(T, \cS')$.
\end{definition}
The smooth sensitivity is a low-sensitivity estimate of the the local
sensitivity. If elements of $\N(\cS)$ have similar local sensitivity, then
$\SS_\beta \approx \LS$. On the other hand, if there are datasets near $\cS$
with local sensitivity close to $\GS$ then $\SS_\beta
\approx \GS$. Formally (see~\cite{nissim2007stoc}), for all $\beta > 0$:
\begin{equation}
\eqlabel{sens-relations}
  \LS(T, \cS) \leq \SS_\beta(T, \cS) \leq \GS(T).
\end{equation}
One crucial property of $\SS_\beta$ is that the Laplace mechanism remains DP if
$\GS$ is replaced by $\SS_\beta$.
\begin{proposition}\cite[Theorem 34]{bun2019nips}
  \proplabel{smoothsens-ed} Fix some $\epsilon, \delta > 0$. Let $\alpha >0$ and
  $\beta >0$ satisfy $\epsilon \geq \alpha + (e^\beta -1) \log(1/\delta)
  - \beta$. Then the mechanism that returns $T(\cS)
  + \frac{\SS_\beta(T, \cS)}{\alpha} Z$, where $Z \sim Lap(0, 1)$, is
  $(\epsilon, \delta)$-DP.
\end{proposition}
It is also possible to achieve $(\epsilon, 0)$-DP with a smooth sensitivity
mechanism.  However, the noise that must be added in this case has polynomial
rather than exponential tail decay.
\begin{proposition}[Combining Lemmas 2.6 and 2.7 from \cite{nissim2007stoc}]
\proplabel{smoothsens-e}
Fix some $\epsilon > 0$ and some $\gamma > 1$.  Define the density
$h(z) \propto \frac{1}{1 + |z|^{\gamma}}$.  Let $\alpha = \beta
= \frac{\epsilon}{2(\gamma + 1)}$.  Then the mechanism that returns $T(\cS)
+ \frac{\SS_\beta(T, \cS)}{\alpha} Z$, where $Z \sim h(z)$, is $(\epsilon,
0)$-DP.
\end{proposition}
Smooth sensitivity is particularly useful for estimating medians or more
generally, quantiles. Indeed, it is not hard to see that the global sensitivity
of the median, for a dataset where elements are in $\output$, is simply
$b-a$. Adding noise of that scale would overwhelm any signal in the output. On
the other hand, for a \emph{usual} dataset, changing a few elements won't change
the median much. Therefore, we would expect smooth sensitivity to be much
smaller than global sensitivity.

We now turn our attention to another popular DP mechanism and the main focus of
this paper: the exponential mechanism (EM). This mechanism defines a
distribution using a \emph{utility function}.
\begin{definition}\cite[Definition 2]{mcsherry2007focs}
  Given a utility function $u: \output \times \bS \to \reals$, the \textbf{exponential mechanism} outputs $x \in \output$
with probability proportional to $\exp\left(\frac{\epsilon u(x, \cS )}{2
  \sensu}\right)$, where $\sensu$ is the sensitivity of $u$:
 $\sensu =  \max_{\cS \in \bS} \max_{\cS' \in \N(\cS)} \max_{x \in \output} |u(x, \cS) - u(x, \cS')|$.
\end{definition}
The EM is always $(\epsilon,0)$-DP. Moreover, it is known to have the
following utility guarantee, whose proof is included in the appendix for
completeness.
\begin{restatable}[\cite{mcsherry2007focs}]{proposition}{expmechprop}
  \proplabel{exponentialmechanismutility}
  Let  $\cS$ be a dataset and $u \colon [a,b] \times \bS \rightarrow \reals$ be a utility function with sensitivity $\sensu$. Let $X $ be a random variable sampled according to the EM and let $\lambda$ denote the Lebesgue measure. If $\textrm{OPT} = \max_x u(x, \cS)$ and $H_t = \{x \colon u(x, \cS) > \textrm{OPT} - t\}$ then:
  \begin{equation*}
    P(u(X, \cS) < \textrm{OPT} - t)\leq \frac{(b-a)}{\lambda(H_{t/2})} e^{-\frac{\epsilon t}{4 \sensu}}.
  \end{equation*}
\end{restatable}
In words: the probability of an output with utility $t$ or more below
$\textrm{OPT}$ is exponentially small in $t$. While the above guarantee is
useful as a general-purpose bound, it is often hard to
interpret. In \secref{ouru} we present the first analysis expressing the utility
guarantees of the exponential mechanism in terms of the more intuitive notion of
smooth sensitivity. This connection:
\begin{itemize}[topsep=0pt,noitemsep,leftmargin=*]
  \item helps us theoretically compare two seemingly unrelated mechanisms, EM
and the smooth-sensitivity-based mechanism of \propref{smoothsens-ed}, and
  \item shows that it is possible to define an $(\epsilon, 0)$-DP mechanism with
faster-than-polynomial noise decay whose variance is also as small as the smooth
sensitivity.
\end{itemize}
This might seem like a contradiction to~\cite{nissim2007stoc}, which states
that \emph{admissible} distributions {\cite[Definition 2.4]{nissim2007stoc}
require polynomial tails.  However, their definition of ``admissible'' does not
encompass all possible DP mechanisms, so it does not actually imply the
non-existence of a DP mechanism with faster-than-polynomial decay.

\section{Related work}
\seclabel{related}

\citet{nissim2007stoc} introduced the notion of smooth sensitivity and
applied it to reduce the noise required to estimate statistics such as
medians. They proposed $(\epsilon, 0)$-DP mechanisms with polynomially decaying
tails and $(\epsilon,\delta)$-DP mechanisms with exponentially decaying
tails. These mechanisms were tightened by~\cite{bun2019nips}, who introduced
light-tailed distributions with variance proportional to the smooth
sensitivity. The privacy guarantees of~\cite{bun2019nips}, however, are under
the relaxed model of \emph{concentrated} DP. Thus, to the best of our knowledge
it remains an open question whether there is a truly $(\epsilon, 0)$-DP
mechanism with variance proportional to smooth sensitivity that also has
faster-than-polynomial tail decay.  In this work we answer the question
affirmatively with \ouru.

In addition to smooth sensitivity, another common approach to tailoring noise to
a dataset is the propose-test-release method~\cite{dwork2009stoc}. This
method tests (in a private way) whether the sensitivity of a function on a
particular dataset is below a given value $\tau$.  If the test passes, then the
mechanism outputs the true function value perturbed by noise with variance
proportional to $\tau$. If the test fails, then \textrm{null} (no output) is
returned. Because there is always a probability of having a false positive in
the test, this mechanism can never achieve $(\epsilon, 0)$-DP. Moreover, the
performance of this mechanism depends heavily on the type of test used on the
first step. Since this dependence is not easily parametrized, designing the best
test for each function is a non-trivial task.

Finally, our work builds upon the well-studied exponential mechanism (EM)
of~\cite{mcsherry2007focs}. Traditionally, the EM has been used for output
spaces where generalizations of the Laplace distribution are not readily
available. For instance, the EM has been exploited when the output space
consists of: databases~\cite{blum2008stoc}, eigenvectors~\cite{eigenvectors},
and infinite-dimensional vectors~\cite{hilbertdp}. In contrast, we propose
applying the EM to outputs, such as medians, which have commonly been made
private using a mechanism involving a Laplace distribution.  To do so, we
introduce a dataset-distance-based utility function family, or \ouru, which
allows us to a) easily control the EM's sensitivity, and b) provide a
data-dependent analysis of the EM. Before this work, the utility guarantees for
the EM were either dataset-independent~\cite{mcsherry2007focs}, specific to a
particular task~\cite{eigenvectors, hilbertdp}, or conditioned on the data being
drawn from a smooth distribution~\cite{wasserman}.

\section{\ouru: dataset-distance-based utility function family}
\seclabel{ouru}
As described in the previous sections, the EM is a popular tool for releasing
private information. However, the particular choice of utility function is
crucial to its performance.  For example, consider the task of estimating the
median of a dataset $\cS = \{y_1, \ldots, y_n\}$, with $y_i \in \output$.  We
will assume the median index $m$ to be $(n+1)/2$ if $n$ is odd and $n/2$ if $n$
is even.  One could use either of the following utility functions:
\begin{equation*}
  u_1(x, \cS) = -|x - T(\cS)| \eqlabel{medianl1}\qquad \textrm{or}\qquad u_2(x, \cS) = -\left|\sum_{y \in \cS} \ind(y < x) - \sum_{y \in \cS} \ind(y > x)\right|.
\end{equation*}
Both functions achieve their maximum value at the true median, $x =
T(\cS)$. However, if we consider their ranges and sensitivities, it becomes
clear that $u_2$ will provide much better median estimates:
\begin{itemize}[topsep=0pt,noitemsep,leftmargin=*]
  \item $u_1$: Range is $[-(b - a), 0]$.  Sensitivity is $\Delta_{u_1} = b-a$;
in the extreme case where there is only a single point in the dataset, $\cS
= \{b\}$, the median can shift from the max value $b$ all the way to the min
value $a$ with the addition of a single new point: $\cS' = \{a, b\}$.

  \item $u_2$: Range is $[-n, 0]$.  Sensitivity is $\Delta_{u_2} = 1$; if a
  point is added to (or removed from) $\cS$, then an indicator function is added
  to (or removed from) each summation in $u_2$, and only one of these two
  indicators will ever be active for a given $x$.
\end{itemize}
The range of $u_1$ is equal to its sensitivity.  In contrast, $u_2$'s range is
much larger than its sensitivity.  This means that the corresponding EM density
is much smaller at quantiles far from the median.  From this example, it is
clear that a good utility function is one that assigns low value to points far
away from the true statistic value, $T(\cS)$, yet has very small sensitivity. We
will now define a dataset-distance-based utility function family (\ouru) that
can achieve this for many real datasets.
\begin{definition}
\deflabel{our-utility} For dataset $\cS$ and statistic $T$, \ouru is
  $\uadmis(x, \cS) = -\min_{\cS' \colon T(\cS') = x} d(\cS,  \cS')$.
\end{definition}
Intuition: This utility function considers all datasets that have statistic
value $x$, and finds one that is closest to the actual dataset $\cS$.  The $d$
value is then a measure of how difficult it is to go from $T(\cS)$ to $x$, and
so its negation is the utility of output $x$ for the dataset $\cS$. One of the important features of $\ouru$ is that the sensitivity of this function is always $1$. (See \lemref{ourusens} in the appendix.)

\subsection{Instantiations}
\seclabel{first-insts}

We now instantiate \ouru for a few statistics to illustrate that it often
reduces to a simple step function.

\paragraph{Modes.} Consider a dataset $\cS$ with values from a
finite set $\Y$. For every $y \in \cS$, let $n_y$ denote the frequency of that
element in $\cS$, and let $n_{\max} = \max_{y \in \cS} n_y$ denote the frequency
of the mode.  Then it is not hard to see that \ouru reduces to the following
simple function: $\uadmis(x, \cS) = n_x - n_{\max}$.

\paragraph{Medians.} Let $\cS= \{y_1, \ldots, y_n\}$ with each
$y_i \in [a, b]$, and assume $y_i < y_{i+1}$.  We
will again consider the median index $m$ to be $(n+1)/2$ if $n$ is odd and $n/2$ if
$n$ is even.  For convenience, we also define $y_0 = a$ and $y_{n+1} = b$.  Then
we can show that $\uadmis$ reduces to the step function given below.
\begin{itemize}[topsep=0pt,noitemsep,leftmargin=*]
\item For $1 \leq k \leq m$ and $y_{m-k} \leq x \leq y_{m-k+1}$, we have $\uadmis(x, \cS) = -2k + \ind(n\;\textrm{odd})$.
\item For $1 \leq k \leq m + \ind(k\;\textrm{even})$ and $y_{m+k-1} \leq x \leq y_{m+k}$, we have $\uadmis(x, \cS) = -2k + \ind(n\;\textrm{even})$.
\end{itemize}

Intuition: Consider the case where $n$ is odd.  By adding an element $y$ between
$y_{m-1}$ and $y_m$, one can obtain a dataset $\cS'$ such that $d(\cS, \cS') = 1$
and $T(\cS') = y$.  The same can be achieved for datasets where $n$ is even by
adding an element $y$ between $y_m$ and $y_{m+1}$.  To move further away from
the original median requires at least two changes to the dataset.  For example,
when $n$ is odd, we can add a new element $y$ between $y_m$ and $y_{m+1}$, then
remove $y_m$, to obtain a dataset $\cS'$ such that $d(\cS, \cS') = 2$ and
$T(\cS') = y$.  A simple inductive argument repeatedly applying the two types of
changes just described yields the utility function form given above.

\paragraph{Means.}
Let $\cS= \{y_1, \ldots, y_n\}$ with each $y_i \in [a, b]$.  Our target
statistic is $\mu \coloneqq \frac{1}{n} \sum_{y \in \cS} y$.  Consider how this statistic changes when points are added to or removed from $\cS$:
\begin{itemize}[topsep=0pt,noitemsep, leftmargin=*]
  \item Adding a single new element $y$: This changes the mean to a value of
  $\mu'(y) \coloneqq \mu\frac{n}{n + 1} + \frac{y}{n + 1}$.  Thus, any mean
  $x$ in the range $\mu'(a) \leq x \leq \mu'(b)$ has utility
  $\uadmis(x, \cS) = -1$ (except for the true mean value, which has $\uadmis(\mu, \cS) =
  0$).

  \item Removing a single element $y_i$: This changes the mean to a value of
  $\mu\frac{n}{n - 1} - \frac{y_i}{n - 1}$.  But since the value of
  each $y_i$ is fixed, the measure of the output space that we can reach with
  removals alone is zero.  Thus, removals are really only meaningful when
  combined with additions.
\end{itemize}
To efficiently determine the minimum combination of additions and removals that
covers the remaining regions of the output space $[a, b]$ requires a dynamic
program.  To see that the combinations that it must consider are bounded, notice
that any mean value is achievable with $n + 1$ changes to the original
data---remove all $n$ original data points and add a point at value $x$ to get
that value as the mean.  Thus, ultimately we end up with a step function that
has no more than $n + 1$ levels.

\subsection{Utility bounds}

We now analyze how similar the output of the EM based on our utility function
$\uadmis$ will be to the true value of a target statistic $T$. Our guarantees
will be stated in terms of the function
$\beta_{T, \cS}^* \colon \reals_+ \rightarrow \reals_+$, defined as the inverse
of the function $\beta \mapsto \frac{\SS_\beta(T, \cS)}{\beta}$. The inverse
$\beta^*_{T, \cS}$ is well-defined since $\frac{SS_\beta(T, \cS)}{\beta}$ is
continuous and is the product of a decreasing and a strictly decreasing
function, which makes it strictly decreasing. When clear from context, we will
remove the dependency on $T$ and $\cS$ from the function $\beta^*$. Proofs of
all theorems, corollaries, etc. can be found in the appendix.
\begin{restatable}{theorem}{thmutility}
  \thmlabel{utility} Let $x \in \output$ denote the output of the
EM with utility function $\uadmis$.  Let $\lambda$
denote the Lebesgue measure and $\gamma =
\frac{1}{2\beta^*(t/(e-1))}$. If $H_t = \{x \mid \uadmis(x, \cS) \geq
-t\}$, then:
  \begin{align}
    P(|x - T(\cS)| > t) \leq \frac{2  \exp\left(-\frac{\epsilon\gamma}{2} \right)(b-a)}{\lambda\left(H_{\gamma}\right)}. \label{eq:error_bound}
    \end{align}
\end{restatable}
As a corollary of the previous theorem we obtain the following
high-probability bound.
\begin{restatable}{corollary}{deltaBoundCoro}{\label{coro:delta-bound}}
Let $H_t$ be as in \thmref{utility} and fix $\eta >0$.  Assume
$\lambda(H_t) \geq C t$ for some constant $C >0$.  Let $\beta_{\exp}
= \frac{\epsilon}{4 W\left(\frac{\epsilon (b-a)}{ C \eta}\right)}$, where $W$
is the main branch of the Lambert function\footnote{The Lambert function is the
inverse of the function $x \mapsto x e^x$}. Then with probability at least
$1-\eta$:
\begin{equation} \eqlabel{exp_guarantee}
|x - T(\cS)| < 4 (e-1) \frac{\SS_{\beta_{exp}}}{\epsilon} W \left(\frac{\epsilon (b-a)}{C \eta}\right).
\end{equation}
\end{restatable}
Our corollary depends on the assumption that $\lambda(H_t) \geq C t$.
  A simple setting where we can estimate the value of $C$ is the following:
Let $\cS$ consist of points evenly-spaced on $[a, b]$ and
consider the problem of estimating the median. Using the \ouru
introduced in \secref{first-insts} we observe that $u_d(x, \cS) > -t$ for
$|x - T(\cS)| < \frac{t(b-a)}{n}$. In this case $C=\frac{2(b-a)}{n}$. In
general, it is hard to estimate $C$, but we believe for many typical datasets
and common statistics it will be $O(\frac{1}{n})$.

We now compare our bound with the one provided in the original smooth
sensitivity paper~\cite{nissim2007stoc}. We use the $(\epsilon, 0)$-DP mechanism
from in \propref{smoothsens-e} with $\gamma = 2$ (yielding the Cauchy
distribution). That is, we set $\beta_{smooth} = \alpha = \frac{\epsilon}{6}$
and release $\T(\cS) + \frac{\SS_{\beta_{smooth}}}{\alpha} Z$ with $Z$ sampled
from density $h(z) = \frac{1}{1 +|z|^2}$. We show in \propref{utilitysmooth} in
the appendix that, for $\eta >0$, the output $x_\SS$ of this mechanism
satisfies:
\begin{equation}
  |x_\SS - T(\cS) | \leq  \frac{6 SS_{\beta_{smooth}}}{\epsilon}
  \tan\left(\frac{\pi(1-\eta)}{2} \right)
  \eqlabel{smoothsensitivity}
\end{equation}
with probability $1 - \eta$. In general this bound is incomparable to
the one for \ouru since $\beta_{smooth}$ could be greater or smaller
than $\beta_{exp}$, depending on the value of $\epsilon$ and
$\eta$. Nonetheless, the dependency on $\eta$ in \eqref{exp_guarantee}
is almost logarithmic. On the other hand as $\eta \rightarrow 0$ it is
not hard to see that $\tan(\frac{\pi(1 - \eta)}{2})$ is in
$\Omega(\frac{1}{\eta})$. Thus we can guarantee that our bound is
exponentially better for small values of $\eta$. 

We want to emphasize that the significance of Corollary~\ref{coro:delta-bound}
is in: 1) establishing, to our knowledge, the first-ever connection between EM
and smooth sensitivity, and 2) showing that there exists an $(\epsilon, 0)$-DP
mechanism with noise on the scale of the smooth sensitivity and
faster-than-polynomial tail decay. Another important thing we want to point out
is that, while the choice of $\beta_{smooth}$ affects not only the utility
guarantees but the privacy of the smooth sensitivity mechanism, for \ouru the
$\beta_{exp}$ is only an artifact of
Corollary~\ref{coro:delta-bound}. Therefore, it might be possible to tighten the
dependency on the smooth sensitivity in
\eqref{exp_guarantee}, whereas \eqref{smoothsensitivity} cannot be
improved without violating privacy.

\section{Sampling algorithms}
\seclabel{algorithms}
One of the main disadvantages of the EM is that, depending on
the shape of the utility function, sampling may not be straightforward or may be
too costly in practice~\cite{eigenvectors, BassilyERM}. In this section we take
advantage of the fact that \ouru takes values in $\mathbb{N}_0$ to provide a
sampling algorithm for it.

Let $A_k = \uadmis^{-1}(-k, \cS) \subset \output$ represent the portion of the output
space that has utility $-k$.  Then let $K \in \mathbb{N}_0$ be a random variable
with the following probability mass function:
\begin{equation}
\eqlabel{pk-defn}
  p_k \coloneqq P(K = k) \propto \lambda(A_k) e^{-\frac{\epsilon k}{2}}.
\end{equation}
Further, let $Z \in \output$ denote a random variable defined by the following
conditional density: $ P(Z = z \mid K = k)
= \frac{1}{\lambda(A_k)}\mathbf{1}(z \in A_k)$. We claim that the distribution
of $Z$ is the one induced by the EM. To verify this, let $z \in \output$ and let
$k_0$ be such that $z \in A_{k_0}$. Then:
\begin{align*}
  P(Z = z) &= \sum_{k\geq 0} P(Z = z \mid K = k) P(K=k) 
 = P(Z = z \mid K = k_0) P(K = k_0)  \\
  &\propto  e^{-\frac{\epsilon k_0}{2}} =
  e^{-\frac{\epsilon \uadmis(z, \cS)}{2}}.
\end{align*}  
Therefore, sampling from the EM is equivalent to sampling an index $k$
proportional to $p_k$ and then sampling uniformly an element from
$A_k$. Sampling uniformly from $A_k$ is typically trivial.  For example, for the
medians estimation task discussed in \secref{first-insts}, each $A_k$ is a
single continuous interval on the real line.  So we focus here on how to sample
$k$ proportional to $p_k$.  In order to solve this problem we need to address
two issues.

First, the domain of $K$ could technically be infinite. In practice however, the
domain of $K$ is often finite.  This is the case for the mode, median, and mean
estimation tasks described in \secref{first-insts}.  In all three of those cases
the set of values \ouru can take is a linear function of the number of elements
in the dataset.  So we focus here on developing an algorithm for the case
where \ouru only takes on a finite number $N$ of values.

The other issue we have to deal with is numeric instability---the probabilities
$p_k$ are numerically equal to $0$ even for moderate values of $k$.  Thus, a
na\"ive sampling of $K$ could ignore the tail values of the distribution. To
address this issue we use a \emph{racing} algorithm\footnote{This algorithm is
from unpublished work by Ilya Mironov, in which he considered the numerical instability issues of estimating medians via the EM.} that samples $N$
random variables and then chooses the smallest one. The algorithm is inspired by
the reservoir sampling algorithm of~\cite{reservoir} and depends on the fact
that the minimum of a collection of exponential random variables also follows an
exponential distribution. The proof that \algmref{sampling} is correct hinges
on \propref{sampling}, which is proven in the appendix.
\begin{restatable}{proposition}{samplingprop}
  \proplabel{sampling}
  Let $Z_k = \log \log \frac{1}{U_k} - \log p_k$ and $K = \argmin Z_k$, where
  $U_k$ are independent and uniformly distributed over $[0,1]$.  Then $P(K = k) \propto
  p_k$.
\end{restatable}

\begin{algorithm}[ht!]
  \caption{Numerically stable sampling}
  \algmlabel{sampling}
  \begin{algorithmic}
    \STATE{\textbf{Input}: Weights  $p_k = \lambda(A_k) e^{-\frac{k
          \epsilon}{2}}$ for $k =0, \ldots, N-1$}
    \STATE{\textbf{Output}: Random variable $K$ s.t. $P(K = k) \propto p_k$}
    \STATE{Sample $N$ uniform r.v.s $U_0, \ldots, U_{N-1}$}
     \STATE{Let $Z_k = \log \log\frac{1}{ U_k} - \log \lambda(A_k) + \frac{
     \epsilon k}{2}$}
     \STATE{Return $\argmin Z_k$}
   \end{algorithmic}
\end{algorithm}

\section{Experiments}
\seclabel{experiments}
In this section we empirically demonstrate the practical advantages of \ouru for
the task of computing medians.  Code to reproduce all experiments is provided in
the supplementary material.

As baselines we use three mechanisms which scale noise according to the smooth
sensitivity.
\begin{itemize}[topsep=0pt,noitemsep,leftmargin=*]
\item \ssmeche: $(\epsilon, 0)$-DP method of \propref{smoothsens-e}, with $\gamma$ set to $2$ (giving the Cauchy distribution).
\item \ssmeched: $(\epsilon, \delta)$-DP method of \propref{smoothsens-ed}. The parameters $\alpha$ and $\beta$ are chosen to minimize $\frac{\SS_\beta} {\alpha}$ under the constraints of \propref{smoothsens-ed}.
\item  \loglap: Recently proposed by~\cite{bun2019nips}, this method
releases $T(S) + \frac{SS_\beta(T, S)}{\alpha} X e^{\sigma Y}$ where
$X \sim \text{Lap}(0,1)$ and $Y \sim \text{N}(0, \sigma)$. Parameters
$\alpha, \beta$, and $\sigma$ are chosen optimally based
on~\cite{bun2019nips}. This mechanism yields
$\frac{\epsilon^2}{2}$-\emph{concentrated} DP.  We use Lemma 9 from~\cite{bun2019nips}
to map to an $(\epsilon, \delta)$-DP guarantee.
\end{itemize}
We also ran preliminary experiments comparing with the propose-test-release
mechanism of~\cite{avella2020propose}, but it was not a strong competitor; even
when allowed a high failure rate (e.g., a return value of ``no result'' for 50\%
of queries), its error on the remaining results was substantially higher than
that of \ouru.

\subsection{Synthetic data}

We fix the dataset size at $|\cS| = 1000$, and generate data from three
distributions: 1) $N(0, 1)$, the zero-mean, unit-variance Normal, 2) $U(0, 1)$,
uniform on $[0, 1]$, and 3) $\mathcal{B}(0.5, 0.5)$, the bimodal Beta.  After
computing the true median, $T(\cS)$, we truncate the data to a reasonable range,
$\output$, effectively imposing limits on the magnitude of user contributions.
In the case of $N(0, 1)$ we set $\output = [-10, 10]$, and otherwise we set
$\output = [0, 1]$.  We note that the exact values of $\output$ are not a
significant factor in the performance of \ouru, but if we loosen these limits
then it dramatically increases the average error of \ssmech and \loglap.  This
is because in many settings a substantial fraction of these methods' probability
mass before clipping to $\output$ lies outside of this range.  Hence, our tight
setting of $\output$ bounds is actually an advantage that competing methods
might not have in a more realistic setting.

We generate $100$ datasets of each type, and call each privacy mechanism $100$
times per dataset.  We compute the average difference between the true median
and the mechanism's estimate, $|T(\cS) - \M(\cS)|$, and average this value for
each dataset.  In \figref{synthetic-exp} we plot the average of these average
errors, with error bars representing the standard deviation across datasets.
(Note that the plots have a log scale on the y-axis, which accounts for the
error bars looking non-symmetric.)

\begin{figure}[ht]
\begin{center}
  \centerline{\includegraphics[width=\textwidth]{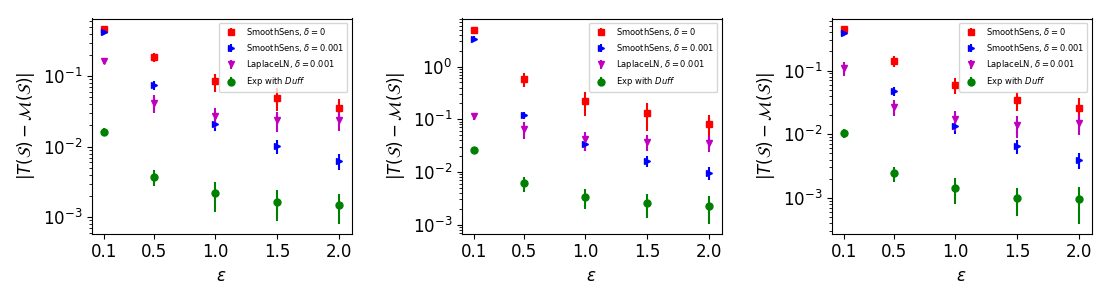}}

  \caption{Average $|T(\cS) - \M(\cS)|$ for various settings of $\epsilon$, with
standard deviation across datasets as error bars.  From left to right, the
data sources for the plots are: $N(0, 1)$, $U(0, 1)$, and $\mathcal{B}(0.5,
0.5)$.} \figlabel{synthetic-exp}


\end{center}
\end{figure}

We observe that \ouru has the smallest average error in all experiments.  It far
outperforms \ssmech in the $\delta=0$ setting; for example, for the $N(0, 1)$
data, the errors of \ouru are a factor of 187 times smaller for $\epsilon=0.1$
and 34 times smaller for $\epsilon=2$.  More remarkably, even if we
allow \ssmech a non-zero $\delta$ value, it still does not reach the performance
of \ouru.  \ouru has errors a factor of 130 ($\epsilon=0.1$) to 4 $(\epsilon=2)$
times smaller than those of \ssmech with the reasonable $\delta$ of $1/|\cS| =
0.001$.  Comparing to the more recent \loglap method, we still see that \ouru
has errors a factor of 4 ($\epsilon=0.1$) to 15 ($\epsilon=2$) times smaller.

Runtimes of \ouru and \ssmech with $\delta = 0$ are similar, and the big-O cost
of all methods is dominated by the required $O(n \log n)$ step of sorting the
data.  However, the optimal parameter search for \ssmech and \loglap with
$\delta > 0$ requires substantial extra time.  In practice, it took
$15$ to $20$ times longer to run these methods.

\subsection{Real data}

\begin{figure*}[!h]
\begin{center}
  \centerline{
  \adjincludegraphics[width=0.3\columnwidth,trim={0 0 {0.83\width} 0},clip]{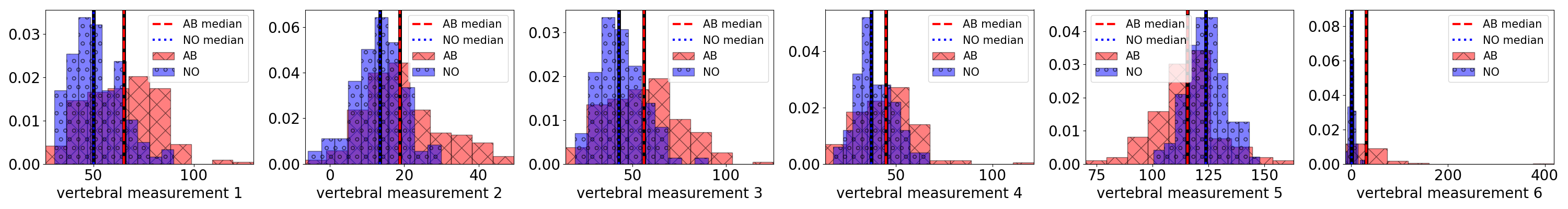}
    \qquad \adjincludegraphics[width=0.3\columnwidth,trim={0 0 {0.83\width} 0},clip]{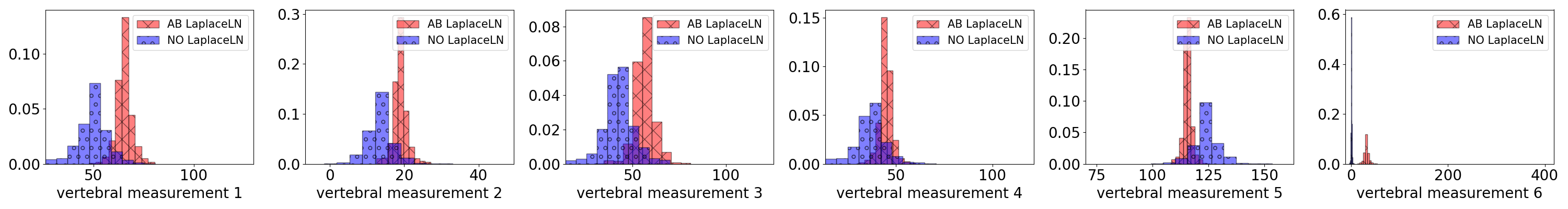}
    \qquad \adjincludegraphics[width=0.3\columnwidth,trim={0 0 {0.83\width} 0},clip]{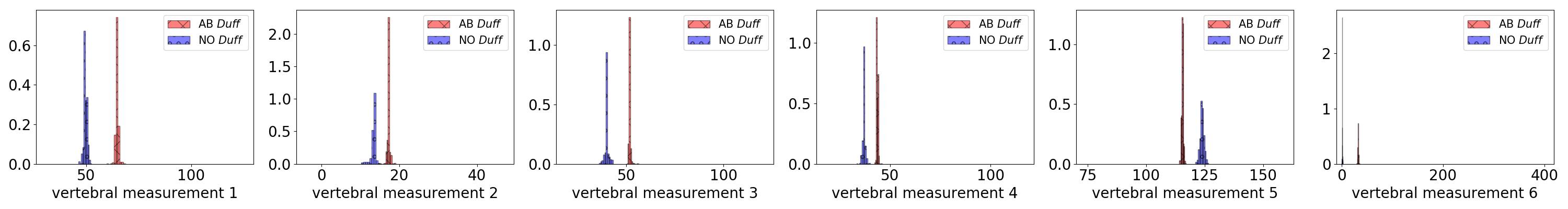}}

  \caption{From left to right, the plots are histograms of: 1) the actual data, 2) median estimates from \loglap with $\delta = 1/n$, and 3) median estimates from the EM with \ouru.} \figlabel{real-exp}

\end{center}
\end{figure*}

As a test of our method on real data, we use the Vertebral Column Dataset from
the UCI Repository~\cite{dua2019uci,mota2019uci}.  Each row in this dataset
contains six measurements related to a patient's vertebral column, and each
patient is classified as either normal (NO) or abnormal (AB).  There are
$n_{NO} = 100$ normal patients and $n_{AB} = 210$ abnormal patients.  The
leftmost plot in \figref{real-exp} shows a normalized histogram of the first
measurement for each class, with class medians denoted by a vertical bar.
Fixing $\epsilon = 0.5$ and setting the $\output$ range to the actual range
observed when the data from both classes is combined, we ran 1000 trials of DP
median estimation. The center and rightmost plots in \figref{real-exp} show
normalized histograms of the resulting estimates from \loglap with $\delta
= 1/n$ and \ouru, respectively.  Notice that the \loglap histograms for the two
classes overlap significantly.  In contrast, the \ouru histograms are much
tighter, providing more useful estimates of the class medians.  Similar plots
for the other five vertebral measurements can be found in the appendix.

\section{Conclusion}
\seclabel{conclusion}
In this work we proposed a general-purpose distance-based utility
function family, \ouru, for the exponential mechanism.  We proved that
\ouru is an affirmative answer to the open question of whether it is
possible to have differential privacy with a noise distribution whose
variance is proportional to smooth sensitivity and whose tails decay
at a faster-than-polynomial rate. This is the first time a connection
between the EM and smooth sensitivity has been studied and we believe
there are very interesting research questions that can expand this
connection..  We also provide the first ever comparison of multiple
median algorithms, demonstrating the advantages of \ouru. An
interesting open question is whether we can instantiate \ouru on other
tasks for which smooth sensitivity constitutes the state of the
art. For instance, the task of privately calculating the number of
triangles in a graph.  Please see below for a broader impact
statement.

\newpage
\section*{Broader impact}

Differential privacy has become a standard method for anonymization.  This work
presents a method that allows organizations to release private statistics such
as the median with much higher fidelity to the true statistic value while at the
same time providing greater user privacy protection.  Notably, in contrast to
the existing state-of-the-art~\cite{bun2019nips}, our method can produce more
accurate results while using no $\delta$.  Since the implications of non-zero
$\delta$ are not well-understood~\cite{mcsherry2017blog}, this may be
especially important for giving private user information the best protection.


\bibliography{biblist}
\bibliographystyle{plainnat}
\setcitestyle{numbers}

\clearpage
\appendix
\section{Additional proofs}
\subsection{Distance properties}
\begin{proposition}
  \proplabel{distance}
  The function $d \colon \bS \times \bS \rightarrow \reals_+$ of
  \defref{dataset-distance} is a metric.
\end{proposition}
\begin{proof}
  The fact that $d$ is positive and symmetric is trivial. Therefore, we focus on
  proving that it satisfies the triangle inequality. Let $\cS$, $\cS'$ and
  $\cS''$ be three data sets. Let $d(\cS, \cS'') = k_1$ and $d(\cS'', \cS') =
  k_2$. Now let $\cS = \cS_0 \ldots, \cS_{k_1} = \cS''$ and $\cS'' = \cS_{k_1} ,
  \ldots, \cS_{k_1 +k_2}= \cS'$ be the sequence of datasets defining $d(\cS,
  \cS'')$ and $d(\cS'', \cS')$. Since $S_0, \ldots, S_{k_1+k_2}$ defines a
  sequence of neighbors starting at $\cS$ and ending at $\cS'$, by definition we
  must have $d(\cS, \cS') \leq k_1 +k_2 = d(\cS, \cS'') + d(\cS'', \cS')$.
\end{proof}

\subsection{Sensitivity of \ouru}
\begin{lemma}
  \lemlabel{ourusens}
  Given any statistic $T$, let $u_d$ be defined as
  \begin{equation*}
    u_d(x, \cS) = -\min_{\cS' \colon T(\cS') = x} d(\cS, \cS').
  \end{equation*}
  Then $\Delta_{u_d}$, the sensitivity of $u_d$, is less than $1$.
\end{lemma}
\begin{proof}
  Let $\cS'' \in \mathcal{N}(\cS)$. We will show that $u_d(x, \cS) \geq u_d(x,\cS'') - 1$, and then the inequality $u_d(x, \cS'') \geq u_d(x, \cS) - 1$ will follow simply by 
symmetry on the choice of $\cS''$ and $\cS$.  The combination of these two implies $u_d(x, \cS'') - 1 \leq u_d(x, \cS) \leq u_d(x, \cS'') + 1$, which yields the statement of the lemma.

Let $S_1 \in \argmin_{\cS' : T(\cS') = x} d(\cS'', \cS')$. Starting from the definition of $u_d$:
  \begin{align*}
    u_d(x, \cS) &= -\min_{\cS' \colon T(\cS') = x} d(\cS, \cS') \\
    &\geq -d(\cS, \cS_1) \\
    & \geq -d(\cS, \cS'') -d(\cS'', \cS_1)\\
    &= -1 + u_d(x, \cS'').
  \end{align*}
  The first inequality follow from the fact that $S_1$ is a feasible dataset for the optimization problem.  The second inequality follows from the triangle inequality on $d$.  The final equality follows from the definition of $\cS_1$ as a minimizer.
  \end{proof}

\subsection{Proof of \propref{exponentialmechanismutility}}

\expmechprop*
      
\begin{proof}
  Let $H_t^c = \{x \colon u(x, \cS) \leq \textrm{OPT} - t\}$, then
  \begin{align*}
    P(u(x, \cS) \leq \textrm{OPT} - t) &= P(H_t^c)  \leq \frac{P(H_t^c)}{P(H_{t/2})} \\
    &=\frac{\int_{H_t^c} e^{\frac{\epsilon u(x, \cS)}{2\sensu}} dx}{\int_{H_{t/2}} e^{\frac{\epsilon u(x, \cS)}{2\sensu}} dx} \\
    &\leq \frac{e^{\frac{\epsilon(\textrm{OPT} - t)}{2 \sensu}}\int_{H_t^c} dx}{e^{\frac{\epsilon(\textrm{OPT - t/2})}{2 \sensu}} \int_{H_{t/2}} dx} \\
    &\leq \frac{e^{-\frac{\epsilon t}{4 \sensu}}(b-a)}{\lambda(H_{t/2})}.
  \end{align*}
\end{proof}

\subsection{Proof of \thmref{utility}}

\thmutility*

\begin{proof}
We begin by bounding the difference between $x$ and $T(\cS)$ in terms of local
sensitivity. Let $\cS^*$ denote a dataset achieving the minimum in the
definition of $\uadmis(x, \cS)$. Let $K$ be a random variable given by $d(\cS,
\cS^*)$.  By definition of $d$, there exists a sequence of neighboring datasets
$\cS = \cS_0, \ldots, \cS_K = \cS^*$. Therefore:
  \begin{align*}
    |x - T(\cS)| = |T(\cS) - T(\cS^*)| 
     = \left|T(\cS) - T(\cS^*) + \sum_{i = 1}^{K-1} T(\cS_i) - \sum_{i = 1}^{K-1} T(\cS_i)\right|
    \end{align*}
Rearranging summands we get:
\begin{align*}
  \left|\sum_{i = 0}^{K-1} (T(\cS_i) - T(\cS_{i+1}))\right|
  \leq \sum_{i=0}^{K-1} |T(\cS_i) - T(\cS_{i+1})|
  \leq \sum_{i=0}^{K-1} \LS(T, \cS_i).
  \end{align*}
where the first inequality follows from the triangle inequality and the second
from the definition of local sensitivity (\defref{local-sens}).
  
We also know from \defref{smooth-sens} that, for any $\beta >0$
and $i \in \naturals$: $\LS(T, \cS_i) \leq e^{\beta i} \SS_\beta(T,
\cS)$. Thus, conditioned on $K \beta \leq 1$, we have:
\begin{align*}
  |x - T(\cS)| & \leq \SS_\beta(T, \cS) \sum_{i=0}^{K-1} e^{\beta i}
                 = \SS_\beta(T, \cS) \frac{e^{K\beta} - 1}{e^\beta -
                 1} \\
 & \leq \frac{(e^{K\beta} -1) \SS_\beta(T, \cS)}{\beta} 
  \leq \frac{\SS_\beta(T, \cS)( K \beta e + (1 - K\beta) -1)}{\beta} \\
  & = K(e - 1)\SS_\beta(T, \cS) 
    = d(\cS, \cS^*)(e - 1)\SS_\beta(T, \cS)
    = -\uadmis(x, \cS)(e - 1) \SS_\beta(T, \cS),
\end{align*}
where the first equality is an application of the geometric summation formula,
and the third inequality follows from convexity of the function $x \mapsto e^x$
and the condition that $K \beta \leq 1$.  We can now bound the error of our
mechanism.
\begin{align}
  P(|x - T(\cS)| > t)  
  &\leq  P(|x - T(\cS) | > t \ \wedge K\beta \leq 1) + P(K\beta > 1)
    \nonumber \\
 &\leq P\left(\uadmis(x, \cS)(e-1)\SS_\beta(T, \cS) < -t\right) 
  + P\left(\uadmis(x, \cS)\beta < -1\right) \nonumber\\
  &=P\left(\uadmis(x, \cS) < \frac{-t}{(e - 1)\SS_\beta(T, \cS)}\right) 
+ P\left(\uadmis(x, \cS) < -\frac{1}{\beta}\right). \eqlabel{prob-bound}
\end{align}
Notice that the above inequality holds for every value of $\beta$ and that the
summands introduce a trade-off. The larger $\beta$ is, the smaller $\SS_\beta$
is, and therefore the smaller the probability of the first event. On the other
hand, a larger value of $\beta$ makes the probability of the second
term larger. Letting $\beta = \beta^*\left(\frac{t}{e-1}\right)$  makes
both events equally likely.  Then \eqref{prob-bound} becomes:
\begin{equation*}
  P(|x-T(\cS)| > t) \leq 2 P\left(\uadmis(x, \cS) < -\frac{1}{\beta^*(t/(e-1))}\right).
\end{equation*}
The result now follows from \propref{exponentialmechanismutility} and the fact
that the sensitivity of $\uadmis$ is $1$: $\Delta_{\uadmis} = 1$.
\end{proof}

\deltaBoundCoro*

\begin{proof}
Using the fact that $\lambda(H_{t/2}) \geq \frac{Ct}{2} $ we can bound the probability in \thmref{utility}:
\begin{equation*}
  P\left(|x - T(\cS)| > t\right) \leq  \frac{2 \exp\left(-\frac{\epsilon \gamma}{2}\right) (b - a)}{C \gamma}
\end{equation*}
Setting the righthand side of the above inequality to $\eta$ and rearranging
terms yields:
\begin{equation*}
\frac{\epsilon \gamma}{2}\exp\left(\frac{\epsilon \gamma}{2}\right) = \frac{\epsilon (b-a)}{C \eta}.
\end{equation*}
Applying the Lambert function, we have:
\begin{equation*}
\frac{\epsilon \gamma}{2} = W\left(\frac{\epsilon (b-a)}{C \eta}\right).
\end{equation*}
Expanding the definition of $\gamma$, this implies:
\begin{equation*}
  \beta^*\left(\frac{t}{e-1}\right) = \frac{\epsilon}{4 W\left(
    \frac{\epsilon (b-a)}{ C \eta}\right)}.
\end{equation*}
Since $\beta^*$ is the inverse function of
$\beta \mapsto \frac{\SS_\beta}{\beta}$, we arrive at the following expression
for $t$:
\begin{equation*}
t = \frac{(e - 1) \SS_{\beta_{exp}}}{\beta_{exp}}.
\end{equation*}
\end{proof}

\subsection{Utility guarantee for the smooth sensitivity mechanism}

\begin{proposition}
  \proplabel{utilitysmooth}
  Let $\gamma > 0$ and let $\mathcal{M}$ be a mechanism that releases $x = T(\cS) + \gamma Z$ where $Z$ is sampled from density $h(z) \propto \frac{1}{1 + |z|^2}$.  Then with probability $1 - \eta$:
  \begin{equation*}
    |x - T(\cS)| < \gamma \tan\left(\frac{\pi}{2}(1 - \eta)\right).
  \end{equation*}
  \end{proposition}
  \begin{proof}
    We begin by calculating the normalization constant of $h$:
    \begin{align*}
      \int_{-\infty}^\infty h(z) dz &= 2 \int_0^\infty \frac{1}{1 + z^2} dz
      = \pi.
    \end{align*}
    We can now measure the error of the mechanism:
    \begin{align*}
      P(|x - T(\cS)| > t) &= P \left(|Z| > \frac{t}{\gamma}\right)\\
      &=\frac{2}{\pi} \int_{\frac{t}{\gamma}}^\infty\frac{1}{1 + z^2}dz \\
      &= 1 - \frac{2}{\pi}\arctan\left(\frac{t}{\gamma}\right)
    \end{align*}
    Setting the righthand side to $\eta$ and solving for $t$ yields:
    \begin{equation*}
      t = \gamma \tan \left(\frac{\pi}{2}(1 - \eta)\right)
      \end{equation*}
  \end{proof}

\subsection{Proof of \propref{sampling}}

\samplingprop*

\begin{proof}
    Let $Z'_k = \frac{\log \frac{1}{U_k}}{p_k}$. From elementary statistics, we know $Z'_k$ is distributed exponential with parameter $p_k$. That is, $P(Z'_k > z) = e^{-p_k z}$.  Now, from \lemref{minexp} we know that $\min_{j \neq k} (Z'_j)$ is distributed exponential with parameter $Q = \sum_{j \neq k} p_j$. Thus, using \lemref{exponential_comp} we have:
    \begin{align*}
      P(K = k) &= P(Z_k \leq \min_j Z_j) = P(Z'_k \leq \min_j Z'_j) \\
               &= P(Z'_k \leq \min_{j\neq k} Z'_j) \\
      &= \frac{p_k}{p_k + Q} = \frac{p_k}{\sum_j p_j}.
    \end{align*}
 \end{proof}
  
\begin{lemma}
  \lemlabel{minexp}
  Let $Z_1, Z_2$ be two independent random variables from exponential distributions with parameters $\lambda_1$ and $\lambda_2$ respectively.  Then  $\min(Z_1, Z_2)$ is an exponential random variable with parameter $\lambda_1 + \lambda_2$.
\end{lemma}

\begin{proof}
  By definition we have:
  \begin{align*}
    P(\min(Z_1, Z_2) \geq z) &= P(Z_1 > z \ \textrm{and} \  Z_2 > z) \\
    &=e^{-\lambda_1 z} e^{-\lambda_2 z} = e^{-(\lambda_1 + \lambda_2) z}.
  \end{align*}
\end{proof}
  
\begin{lemma}
\lemlabel{exponential_comp}
Let $Z_1, Z_2$ be two independent random variables from exponential distributions with parameters $\lambda_1$ and $\lambda_2$ respectively.  Then:
\begin{equation*}
  P(Z_1 \leq Z_2) = \frac{\lambda_1}{\lambda_1 + \lambda_2}
\end{equation*}
\end{lemma}

\begin{proof}
  \begin{align*}
    P(Z_1 \leq Z_2) &= \int_0^\infty \int_{z_1}^\infty \lambda_1 \lambda_2 e^{-\lambda_1 z_1}e^{-\lambda_2 z_2} dz_2 dz_1 \\
   &= \int_0^\infty  -\lambda_1 e^{-(\lambda_1 + \lambda_2) z_1} dz_1\\
    &= \frac{\lambda_1}{\lambda_1+ \lambda_2}.
  \end{align*}
\end{proof}

\section{Additional plots for real data experiments}

\begin{figure*}[ht]
\begin{center}
  \adjincludegraphics[width=\columnwidth,trim={{0.17\width} 0 0 0},clip]{figures/vert-data-exp.png}
  \adjincludegraphics[width=\columnwidth,trim={{0.17\width} 0 0 0},clip]{figures/vert-ll-exp.png}
  \adjincludegraphics[width=\columnwidth,trim={{0.17\width} 0 0 0},clip]{figures/vert-duff-exp.png}
  \caption{From top to bottom, the plots are histograms of: 1) the actual data,
    2) median estimates from \loglap with $\delta = 1/n$, and 3) median
    estimates from the EM with \ouru.}
  \figlabel{extra-real-exp}

\end{center}
\end{figure*}

\end{document}